\documentclass[aps,preprintnumbers,nofootinbib,superscriptaddress,twocolumn]{revtex4}
\usepackage{epsfig,graphics}
\usepackage{amsfonts,amssymb,amsmath}            
\usepackage{amsthm}

\newcommand{\ot}{\otimes}
\newcommand{\comment}[1]{}
\newcommand{\ket}[1]{|#1\rangle}

\newcommand{\ketbra}[2]{|#1\rangle\!\langle#2|}

\newcommand*{\tr}{\mathrm{tr}}

\newcommand*{\cE}{\mathcal{E}}
\newcommand*{\cH}{\mathcal{H}}

\newcommand*{\bcE}{\mathcal{\bar{E}}}
\newcommand{\aSV}{SV}
\newcommand{\aF}{FR}
\newcommand{\aQ}{QM}
\newcommand{\aQa}{QMa}
\newcommand{\aQb}{QMb}

\newcommand{\bef}{\rightsquigarrow} 
\newcommand*{\nbef}{\not\rightsquigarrow}
\newcommand*{\extension}{\Xi} 

\def\bibsection{\section*{REFERENCES}} 

\theoremstyle{plain}

\newtheorem{lemma}{Lemma}

\theoremstyle{definition}
\newtheorem{definition}{Definition}

\begin{document}

\title{No extension of quantum theory can have improved predictive
  power\footnote{The published version of this work can be found in
    {\it Nature Communications}~{\bf 2}, 411 (2011) at \url{http://www.nature.com/ncomms/journal/v2/n8/full/ncomms1416.html}.}}

\date{9th August 2011}

\author{Roger \surname{Colbeck}}
\email[]{rcolbeck@perimeterinstitute.ca}
\affiliation{Perimeter Institute for Theoretical Physics, 31 Caroline Street North, Waterloo, ON N2L 2Y5, Canada}
\author{Renato \surname{Renner}}
\affiliation{Institute for Theoretical Physics, ETH Zurich, 8093
Zurich, Switzerland}

\begin{abstract}
  According to quantum theory, measurements generate random outcomes,
  in stark contrast with classical mechanics.  This raises the
  question of whether there could exist an extension of the theory
  which removes this indeterminism, as suspected by Einstein, Podolsky
  and Rosen (EPR).  Although this has been shown to be impossible,
  existing results do not imply that the current theory is maximally
  informative.  Here we ask the more general question of whether
  \emph{any} improved predictions can be achieved by \emph{any}
  extension of quantum theory.  Under the assumption that measurements
  can be chosen freely, we answer this question in the negative: no
  extension of quantum theory can give more information about the
  outcomes of future measurements than quantum theory itself.  Our
  result has significance for the foundations of quantum mechanics, as
  well as applications to tasks that exploit the inherent randomness
  in quantum theory, such as quantum cryptography.
\end{abstract}

\maketitle

Given a system and a set of initial conditions, classical mechanics
allows us to calculate the future evolution to arbitrary precision.
Any uncertainty we might have at a given time is caused by a lack of
knowledge about the configuration.  In quantum theory, on the other
hand, certain properties|for example position and momentum|cannot both
be known precisely.  Furthermore, if a quantity without a defined
value is measured, quantum theory prescribes only the probabilities
with which the various outcomes occur, and is silent about the
outcomes themselves.

This raises the important question of whether the outcomes could be
better predicted within a theory beyond quantum mechanics~\cite{EPR}.
An intuitive step towards its answer is to consider appending
\emph{local hidden variables} to the theory~\cite{Bell}.  These are
classical variables that allow us to determine the experimental
outcomes (see later for a precise definition).  Here we ask a new,
more general question: is there \emph{any} extension of quantum theory
(not necessarily taking the form of hidden variables) that would
convey \emph{any} additional information about the outcomes of future
measurements?

We proceed by giving an illustrative example.  Consider a particle
heading towards a measurement device which has a number of possible
settings, denoted by a parameter, $A$, corresponding to the different
measurements that can be chosen by the experimenter.  The measurement
generates a result, denoted $X$.  For concreteness, one could imagine
a spin-$\frac{1}{2}$ particle incident on a Stern-Gerlach apparatus.
Each choice of measurement corresponds to a particular orientation of
the device and the outcome is assigned depending on which way the beam
is deflected.  Within quantum theory, a description of the quantum
state of the particle and of the measurement apparatus allows us to
calculate the distribution, $P_{X|A}$, of the outcome, $X$, for each
measurement choice, $A$.  Another example is described in
Figure~\ref{fig:setup}.

\begin{figure}
\includegraphics[width=0.5\textwidth]{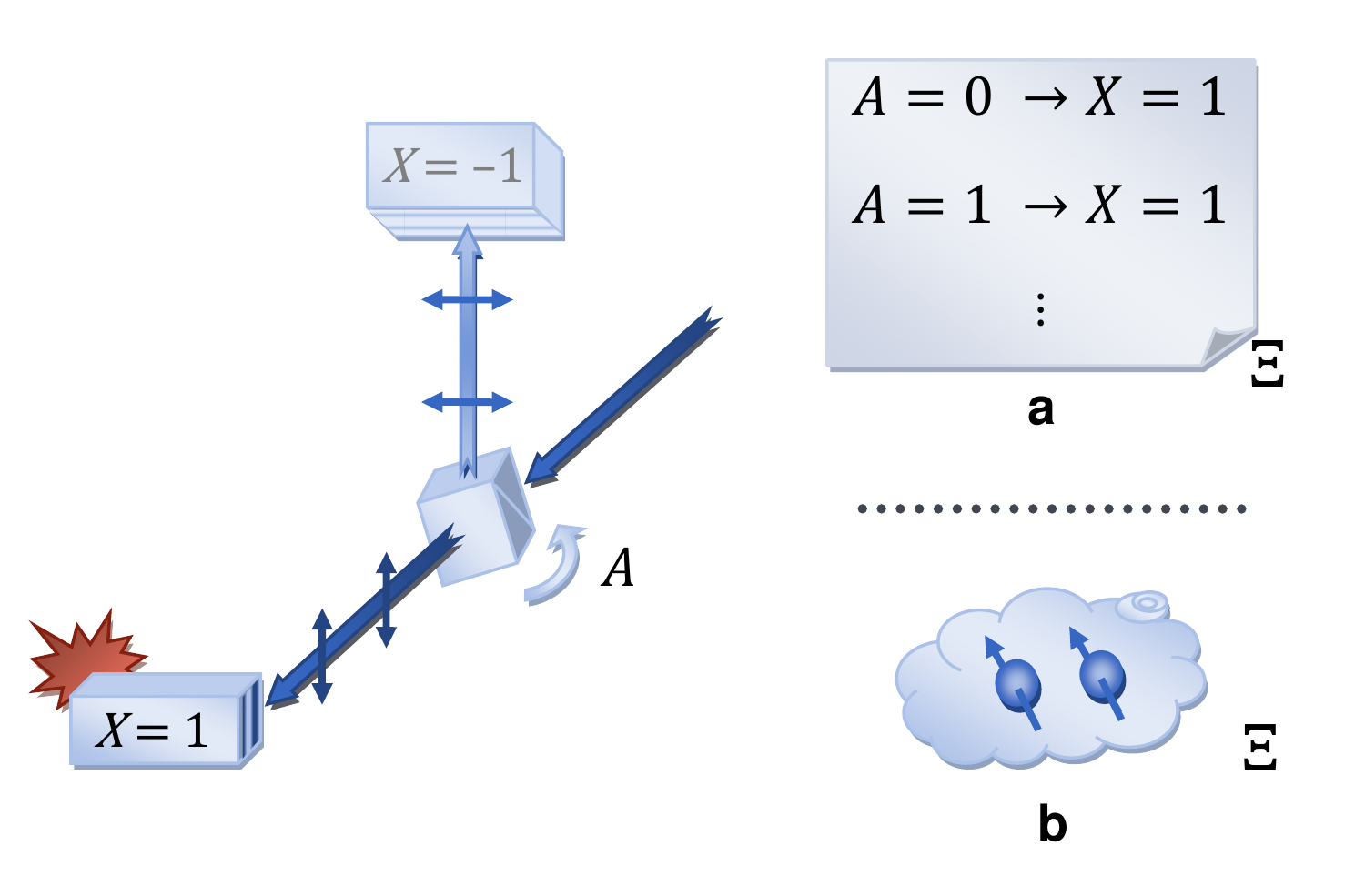}
  \caption{{\sf \textbf{Illustration of the scenario.} A measurement
      is carried out on a particle, depicted as a photon measured
      using an arrangement comprising a polarizing beam splitter and
      two detectors.  The measurement choice (the angle of the
      polarizing beam splitter) is denoted $A$ and the outcome, $X$,
      is assigned $-1$ or $1$ depending on which detector fires.  On
      the right, we represent the additional information that may be
      provided by an extended theory, $\extension$, shown here taking
      the form of either \textbf{a}~hidden variables, i.e., a
      classical list assigning outcomes, or \textbf{b}~a more general
      (e.g.\ quantum) system.}}
\label{fig:setup}
\end{figure}

In this work we consider the possibility that there exists additional,
yet to be discovered, information that allows the outcome $X$ to be
better predicted.  We do not place any restrictions on how this
information is manifest, nor do we demand that it allows the outcomes
to be calculated precisely.  In particular, it could be that the
additional information gives rise to a more accurate distribution over
the outcomes.  For example, in an experiment for which quantum theory
predicts a uniform distribution over the outcomes, $X$, there could be
additional information that allows us to calculate a value, $X'$, such
that $X=X'$ with probability $\frac{3}{4}$ (in the model proposed by
Leggett~\cite{leggett}, for instance, the local hidden variables
provide information of this type).  More generally, we allow for the
possibility of an extended theory that provides non-classical
information.  For example, it could comprise a ``hidden quantum
system'', which, if measured in the correct way, gives a value
correlated to~$X$.

\section{Results}
\subsection*{Assumptions}
In order to formulate our main claim about the non-extendibility of
quantum theory, we introduce a framework within which any arbitrary
additional information provided by an extension of the current theory
can be considered.  In the following, we explain this framework on an
informal level (see the Supplementary Information for a formal
treatment).

The crucial feature of our approach is that it is operational, in the
sense that we only refer to directly observable objects (such as the
outcome of an experiment), but do not assume anything about the
underlying structure of the theory.  Note that the outcome, $X$, of a
measurement is usually observed at a certain point in spacetime. The
coordinates of this point (with respect to a fixed reference system)
can be determined operationally using clocks and measuring
rods. Analogously, the measurement setting $A$ needs to be available
at a certain spacetime point (before the start of the experiment). To
model this, we introduce the notion of a \emph{spacetime random
  variable} (\aSV), which is simply a random variable together with
spacetime coordinates $(t, r_1, r_2, r_3)$. Operationally, a \aSV{}
can be interpreted as a value that is accessible at a given spacetime
point $(t, r_1, r_2, r_3)$.  We now model a measurement process as one
that takes an input, $A$, to an output, $X$, where both $X$ and $A$
are \aSV{}s.

Our result is based on the assumption that measurement settings can be
chosen freely (which we call Assumption~\emph{\aF{}}).  We note that
this assumption is common in physics, but often only made implicitly.
It is, for example, a crucial ingredient in Bell's theorem
(see~\cite{Bell_free}).  Formulated in our framework,
Assumption~\emph{\aF{}} is that the input, $A$, of a measurement
process can be chosen such that it is uncorrelated with certain other
\aSV{}s, namely all those whose coordinates lie outside the future
lightcone of the coordinates of $A$.  We note that this reference to a
lightcone is only used to identify a set of \aSV{}s, and does not
involve any assumptions about relativity theory (see the Supplementary
Information).  However, the motivation for Assumption~\emph{\aF{}} is
that, when interpreted within the usual relativistic spacetime
structure, it is equivalent to demanding that $A$ can be chosen such
that it is uncorrelated with any \emph{pre-existing} values in any
reference frame.  That said, the lack of correlation between the
relevant \aSV{}s could be justified in other ways, for example by
using a notion of ``effective freedom'' (discussed
in~\cite{Bell_free}).

We also remark that Assumption~\emph{\aF{}} is consistent with a
notion of relativistic causality in which an event $B$ cannot be the
\emph{cause} of $A$ if there exists a reference frame in which $A$
occurs before $B$.  In fact, our criterion for $A$ to be a \emph{free
  choice} is satisfied whenever anything correlated to $A$ could
potentially have been caused by $A$. However, in an alternative world
with a universal (frame-independent) time, one might reject
Assumption~\emph{\aF{}} and replace it with something weaker, for
example that $A$ is free if it is uncorrelated with anything in the
past with respect to this universal time.  Nevertheless, since
experimental observations indicate the existence of relativistic
spacetime, we use a notion of \emph{free choice} consistent with this.

We additionally assume that the present quantum theory is correct (we
call this Assumption~\emph{\aQ{}}).  This assumption is natural since
we are asking whether quantum theory can be \emph{extended}.  In fact,
we only require that two specific aspects of quantum theory hold, and
so split Assumption~\emph{\aQ{}} into two parts.  On an informal
level, the first is that measurement outcomes obey quantum statistics,
and the second is that all processes within quantum theory can be
considered as unitary evolutions if one takes into account the
environment (see the Supplementary Information for more details).  We
remark that the second part of this assumption need only hold for
microscopic processes on short timescales and does not preclude
subsequent wave function collapse.

\subsection*{Main Findings}
Consider a measurement which depends on a setting $A$ and produces an
output $X$.  According to quantum theory, we can associate a quantum
state and measurement operators with this process from which we can
compute the distribution $P_{X|A}$.

We ask whether there could exist an extension of quantum theory that
provides us with additional information (which we denote by
$\extension$) that is useful to predict the outcome.  In order to keep
the description of the information, $\extension$, as general as
possible, we do not assume that it is encoded in a classical system,
but instead characterize it by how it behaves when
observed. (Formally, we model access to $\extension$ analogously to
the measurement of a quantum system, i.e., as a process which takes an
input \aSV{} and produces an output \aSV{}.)  We demand that
$\extension$ can be accessed at any time (similarly to classical or
quantum information held in a storage device) and that it is
\emph{static}, i.e., its behaviour does not depend on where or when it
is observed.

Our main result is that we answer the above question in the negative,
i.e., we show that, using Assumptions~\emph{\aF{}} and~\emph{\aQ{}},
the distribution $P_{X|A}$ is the most accurate description of the
outcomes.  More precisely, for any fixed (pure) state of the system,
the chosen measurement setting, $A$, is the only non-trivial
information about $X$, and any additional information, $\extension$,
provided by an extended theory is irrelevant.  We express this via the
Markov chain condition
\begin{equation}\label{eq:Markov}
X\leftrightarrow A\leftrightarrow \extension\, .
\end{equation}
This condition expresses mathematically that the distribution of $X$
given $A$ and $\extension$ is the same as the distribution of $X$
given only $A$~\cite{CovThoMarkov}.  Hence, access to $\extension$
does not decrease our uncertainty about $X$, and there is no better way
to predict measurement outcomes than by using quantum theory.

In the Methods, we sketch the proof of this (the full proof is
deferred to the Supplementary Information).

\section{Discussion}
We now discuss experimental aspects related to our result.  Note that
at the formal level, we present a theorem about certain defined
concepts based on certain assumptions, hence what remains is to
connect our definitions to observations in the real world, and
experimentally confirm the assumptions, where possible.
Assumption~\emph{\aF}\ refers to the ability to make free choices
and|while we can never rule out that the universe is deterministic and
that free will is an illusion|this is in principle falsifiable, e.g.\
by a device capable of guessing an experimentalist's choices before
they are made.  (See also~\cite{CR_free} where the possibility of
weakening this assumption is discussed.)

The validity of Assumption~\emph{\aQ{}}\ could be
argued for based on experimental tests of quantum theory.  However,
the existence of the particular correlations we use in the second part
of our proof is quantum-theory independent, so worth establishing
separately.  Due to experimental inefficiencies, these correlations
cannot be verified to arbitrary precision.  Figure~\ref{fig:visib}
bounds our ability to experimentally establish~\eqref{eq:Markov}
depending on the quality of the setup used (characterized here by the
visibility).  For more details, see the Methods.

\bigskip

We proceed by discussing previous work on extensions of quantum
theory.  To the best of our knowledge, all such extensions that have
been excluded to date can also be excluded using our result.

The question asked by EPR~\cite{EPR} was whether quantum mechanics
could be considered complete.  They appealed to intuition to argue
that an extended theory should exist and one might then have hoped for
a \emph{deterministic completion}, i.e.\ one that would uniquely
determine the measurement outcomes|contrast this with our (more
general) notion, where the extended theory may only give partial
information.  Bell~\cite{Bell} famously showed that a deterministic
completion is not possible when the theory is supplemented by
\emph{local hidden variables}.  (To relate this back to our result,
this corresponds to the special case where the additional information,
$\extension$, is a classical value specified by the local hidden
variables.  A short discussion on the term \emph{local} can be found
in the Supplementary Information.)  Recently, a
conclusion~\cite{ConKoc_FW1} similar to Bell's has been reached using
the Kochen-Specker theorem~\cite{KS}.  These results have been
extended to arbitrary (i.e.\ not necessarily local) hidden
variables~\cite{gisin_covariant, blood} under the assumption of
relativistic covariance (see also~\cite{CRcomment}, as well as
\cite{tresser} where a condition slightly weaker than locality is used
to derive a theorem similar to Bell's).

The aforementioned papers left open the question of whether there
could exist an extended theory which provides additional information
about the outcomes without determining them completely.  (Note that,
in his later works, Bell uses definitions that potentially allow
probabilistic models~\cite{Bell_nouvelle}.  However, as explained in
the Supplementary Information, non-deterministic models are not
compatible with Bell's other assumptions.)  In the case that the
additional information takes the form of \emph{local hidden
  variables}, an answer to the above question can be found
in~\cite{leggett,BBGKLLS,ColbeckRenner}, and the strongest result is
that any local hidden variables are necessarily uncorrelated with the
outcomes of measurements on Bell states~\cite{ColbeckRenner}. (We
remark that the model in~\cite{leggett} also included non-local hidden
variables.  However, we have not referred to these in this paragraph,
since, as mentioned below in the context of de Broglie-Bohm theory,
the presence of non-local hidden variables contradicts
Assumption~\emph{\aF{}}.)

In the present work, we have taken this idea further and excluded the
possibility that \emph{any} extension of quantum theory (not
necessarily in the form of local hidden variables) can help predict
the outcomes of \emph{any} measurement on \emph{any} quantum state.
In this sense, we show the following: under the assumption that
measurement settings can be chosen freely, quantum theory really is
complete.

We remark that several other attempts to extend quantum theory have
been presented in the literature, the de Broglie-Bohm
theory~\cite{deBroglie,Bohm} being a prominent example (this model
recreates the quantum correlations in a deterministic way but uses
non-local hidden variables, see e.g.~\cite{Bell_Bohm} for a summary).
Our result implies that such theories necessarily come at the expense
of violating Assumption~\emph{\aF}.

Another way to generate candidate extended theories is via models
which simulate quantum correlations.  We discuss the implications
of our result in light of such models in the Supplementary
Information.  In addition, we remark that a claim in the same spirit
as ours has recently been obtained based on the assumption of
non-contextuality~\cite{ChenMontina}.

\bigskip

Randomness is central to quantum theory and with it comes a range of
philosophical implications.  In this Article we have shown that the
randomness is inherent: any attempt to better explain the outcomes of
quantum measurements is destined to fail.  Not only is the universe
not deterministic, but quantum theory provides the ultimate bound on
how unpredictable it is.  Aside from these fundamental implications,
there are also practical ones.  In quantum cryptography, for example,
the unpredictability of measurement outcomes can be quantified and
used to restrict the knowledge of an adversary.  Most security proofs
implicitly assume that quantum theory cannot be extended (although
there are exceptions, the first of which was given in~\cite{BHK}).
However, in this work, we show that this follows if the theory is
correct.

\section{Methods}
Our main result is the following theorem whose proof we sketch here
(see the Supplementary Information for the formal treatment).

\bigskip

\noindent \emph{Theorem 1.}|For any quantum measurement with input
\aSV{} $A$ and output \aSV{} $X$ and for any additional information,
$\extension$, under Assumptions~\emph{\aQ} and~\emph{\aF}, the Markov
chain condition~\eqref{eq:Markov} holds.

\bigskip

The proof is divided into three parts.  The first two are related to a
Bell-type setting, involving measurements on a maximally entangled
state.  In Part~I, we show that Assumption~\emph{\aF}\ necessarily
enforces that $\extension$ is \emph{non-signalling} (in the sense
defined below).  In Part~II we show that for a particular set of
bipartite correlations, if $\extension$ is non-signalling, it cannot
be of use to predict the outcomes.  These correlations occur in
quantum theory (cf.\ the first part of Assumption~\emph{\aQ{}}) when
measuring a maximally entangled state and hence we conclude that no
$\extension$ can help predict the outcomes of measurements on one half
of such a state.  Finally, in Part~III, we use the second part of
Assumption~\emph{\aQ{}} to argue that this conclusion also applies to
all measurements on an arbitrary (pure) quantum state.  Together,
these establish our claim.

The bipartite scenario used for the first two parts of the proof
involves two quantum measurements, with inputs $A$ and $B$ and
respective outputs $X$ and $Y$.  The setup is such that the two
measurements are spacelike separated in the sense that the coordinates
of $A$ are spacelike separated with the coordinates of $Y$, and,
likewise, those of $B$ are spacelike separated with those of~$X$.

As mentioned in the main text, we model the information provided by
the extended theory, $\extension$, by its behaviour under observation.
We introduce a~\aSV{}, $C$, which can be thought of as the choice of
what to observe, and another~\aSV{}, $Z$, which represents the outcome
of this observation.  In terms of these variables, our main result,
Equation~\eqref{eq:Markov}, can be restated that for all values of
$a$, $c$ and $x$, we have 
\begin{align}
  P_{Z|acx}=P_{Z|ac} \ .  \label{eq:Markalt}
\end{align}
(Note that we use lower case to denote specific values of the
corresponding upper case \aSV{}s.)

\bigskip

\emph{Proof: Part~I.}|The entire setup described above (including the
additional information $\extension$, accessed by choosing an
observable, $C$, and obtaining an outcome, $Z$) gives rise to a joint
distribution $P_{XYZ|ABC}$.  The purpose of this part of the proof is
to show that Assumption~\emph{\aF}\ implies that $P_{XYZ|ABC}$ must
satisfy particular constraints, called \emph{non-signalling
  constraints}, which characterize situations where operations on
different isolated systems cannot affect each other.  Formally, these
are 
\begin{align}
P_{YZ|ABC} & =P_{YZ|BC}  \label{eq:nsone} \\
P_{XZ|ABC} & =P_{XZ|AC} \label{eq:nstwo} \\
P_{XY|ABC} & =P_{XY|AB}
\end{align}
We remark that the observation that the assumption of free choice
gives rise to certain non-signalling constraints has been made already
in~\cite{CRcomment}, and a similar argument has been presented by
Gisin~\cite{gisin_covariant} and Blood~\cite{blood}. (Note that the
arguments in~\cite{gisin_covariant, blood} implicitly assume that
measurements can be chosen freely).

Assumption~\emph{\aF}\ allows us to make $A$ a free choice and hence
we have 
\begin{align} \label{eq:freechoice}
  P_{A|BCYZ}=P_A
\end{align}
(the setup is such that the measurements specified by $A$ and $B$ are
spacelike separated and, furthermore, $\extension$ is static, so we
can consider the case where its observation is also spacelike
separated from the measurements specified by $A$ and
$B$). Furthermore, using the definition of conditional probability
($P_{Q|R} := P_{Q R} / P_R$), we can write
\begin{align*}
P_{YZA|BC}
=  P_{YZ|BC}\times P_{A|BCYZ}
=  P_A\times P_{YZ|BC}  \ ,
\end{align*}
where we inserted~\eqref{eq:freechoice} to obtain the second
equality. Similarly, we have
\begin{align*}
P_{YZA|BC}
= P_{A|BC}\times P_{YZ|ABC}
= P_A\times P_{YZ|ABC} \ .
\end{align*}
Comparing these two expressions for $P_{YZA|BC}$ yields the desired
non-signalling condition~\eqref{eq:nsone}.  By a similar argument the
other non-signalling conditions can be inferred from
Assumption~\emph{\aF}.

\bigskip

\emph{Proof: Part~II.}|For the second part of the proof, we consider
the distribution $P_{XY|AB}$ resulting from certain appropriately
chosen measurements on a maximally entangled state.  We show that any
enlargement of this distribution (via a system $\extension$ that is
accessed in a process with input \aSV{} $C$ and output \aSV{} $Z$) to
a distribution $P_{XYZ|ABC}$ which satisfies the above non-signalling
conditions is necessarily trivial in the sense that $\extension$ is
uncorrelated to the rest.  For this we draw on ideas from
\emph{non-signalling cryptography}~\cite{BHK}, which are related to
the idea of basing security on the violation of Bell
inequalities~\cite{Ekert}.  Technically, we employ a lemma (see
Lemma~1 in the Supplementary Information), whose proof
is based on \emph{chained Bell inequalities}~\cite{pearle,BC} and
generalizes results of~\cite{BKP,ColbeckRenner}.

Consider any bipartite measurement with inputs $A \in
\{0,2,\ldots,2N-2\}$ and $B\in\{1,3,\ldots,2N-1\}$, for some $N \in
\mathbb{N}$, and binary outcomes, $X$ and $Y$.  The correlations of
the outcomes can be quantified by
\begin{align} \label{eq:INdef}
I_N := P(X = Y | A=0, B=2 N -1)\ +
\\\nonumber\sum_{\genfrac{}{}{0pt}{}{a,b}{|a-b|=1}} P(X \neq Y |A=a, B=b)  \ .
\end{align}
Our lemma then asserts that, under the non-signalling conditions
derived in Part~I, 
\begin{align}
D(P_{Z|abcx},P_{Z|abc})&\leq I_N \label{EQ:DBOUND}
\end{align}
for all $a$, $b$, $c$ and $x$, where $D$ is the \emph{variational
  distance}, defined by
$D(P_Z,Q_Z):=\frac{1}{2}\sum_z|P_Z(z)-Q_Z(z)|$.  The variational
distance has the following operational interpretation: if two
distributions have variational distance at most~$\delta$, then the
probability that we ever notice a difference between them is at
most~$\delta$.

The argument up to here is formally independent of quantum theory.
However, as we describe below (see the Experimental Verification
section), for any fixed orthogonal rank-one measurement on a two-level
subsystem, one can construct $2N-1$ other measurements such that,
according to quantum theory, applying these measurements to maximally
entangled states leads to correlations which satisfy
$I_N\propto\frac{1}{N}$.  It follows that, in the limit of large $N$,
an arbitrarily small bound on $D(P_{Z|abcx},P_{Z|abc})$ can be
obtained. We thus conclude that $P_{Z|abcx} = P_{Z|abc}$, which, by
the non-signalling condition~\eqref{eq:nstwo}, also
implies~\eqref{eq:Markalt}. We have therefore shown that the
relation~\eqref{eq:Markov} holds for the outcome $X$ of any orthogonal
rank-one measurement on a system that is maximally entangled with
another one (our claim can be readily extended to systems of dimension
$2^t$ for positive integer $t$ by applying the result to $t$ two-level
systems).

We also remark that Markov chains are reversible, i.e.\
$P_{Z|abcx}=P_{Z|abc}$ implies $P_{X|abcz}=P_{X|abc}$, which together
with the non-signalling conditions gives $P_{X|abcz}=P_{X|a}$.  This
establishes that, for any choices of $B$ and $C$, learning $Z$ does
not allow an improvement on the quantum predictions, $P_{X|a}$.

\bigskip

\emph{Proof: Part~III.}|To complete our claim, it remains to show that
the Markov chain condition~\eqref{eq:Markov} holds for measurements on
arbitrary states (not only for those on one part of a maximally
entangled state shared between two sites). The proof of this proceeds
in two steps. The first is to append an additional measurement with
outcome $X'$, chosen such that the pair $(X,X')$ is uniformly
distributed. In the second step, we split the measurement into two
conceptually distinct parts, where, in the first, the measurement
apparatus becomes entangled with the system to be measured (and,
possibly the environment) and, in the second, this entangled state is
measured giving outcomes $(X,X')$. Since these outcomes are uniformly
distributed, the state before the measurement can be considered
maximally entangled, so that~\eqref{eq:Markov} holds with $X$ replaced
by $(X,X')$.  This implies~\eqref{eq:Markov} and hence completes the
proof of Theorem~1.

\bigskip

\begin{figure}
\includegraphics[width=0.45\textwidth]{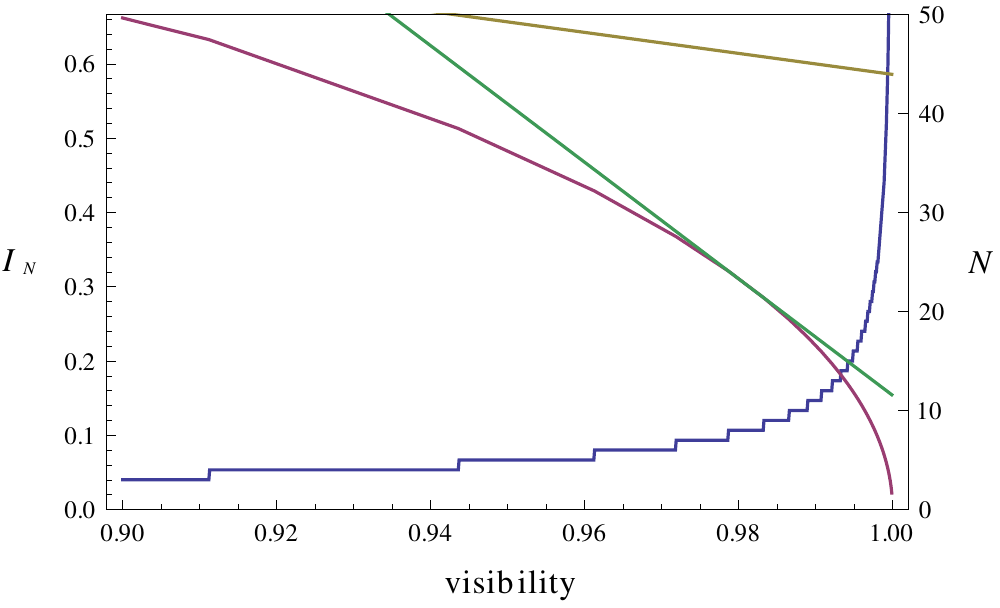}
\caption{{\sf \textbf{Achievable values of $I_N$ depending on the
      experimental visibility.} This figure relates to the measurement
    setup used for testing the accuracy of Assumption~\emph{\aQ{}} as
    described in the Methods.  The setup involves two parties and is
    parameterized by the number of possible measurement choices
    available to each party, $N$.  The plot gives the minimum $I_N$
    achievable depending on the visibility (red line), which
    determines the smallest upper bound on the variational distance
    from the perfect Markov chain condition~\eqref{eq:Markov} that
    could be obtained with that visibility (see
    Eq.~\eqref{EQ:DBOUND}).  It also shows the optimal value of $N$
    which achieves this (blue line).  For comparison, the values
    achievable using $N=2$, which corresponds to the CHSH
    measurements~\cite{CHSH} (yellow line), and the case $N=8$, which
    is optimal for visibility $0.98$ (green line), are shown.}}
\label{fig:visib}
\end{figure}

\emph{Experimental Verification}|As explained above, the validity of
parts of Assumption~\emph{\aQ{}} can be established by a direct
experiment.  In particular, to verify the existence of the
correlations required for Part~II of the proof, i.e.\ those with small
$I_N$, one should generate a large number (much larger than $N$) of
maximally entangled particles and distribute them between the
measurement devices.  At spacelike separation, a two-level subsystem
(e.g.\ a spin degree of freedom) should then be measured, the
measurement being picked at random from those specified below, and the
results recorded.  This is repeated for all of the particles.  The
measurement choices and results are then collected and used to
estimate the terms in $I_N$ using standard statistical techniques.

For an arbitrary orthogonal basis $\{\ket{0}, \ket{1}\}$, the required
measurements can be constructed in the following way.  Recall that the
choice of measurement on one side takes values $A\in\{0,2,\ldots,
2N-2\}$ and similarly, $B\in\{1,3,\ldots, 2N-1\}$.  We define a
set of angles $\theta^j=\frac{\pi}{2 N}j$ and states
$$\{\ket{\theta^j_{+}},\ket{\theta^j_{-}}\}\!=\!\left\{\cos\!\frac{\theta^j}{2}\ket{0}+\sin\!\frac{\theta^j}{2}\ket{1},\sin\!\frac{\theta^j}{2}\ket{0}-\cos\!\frac{\theta^j}{2}\ket{1}\!\right\}\!.$$
The required measurement operators are then
$E_{\pm}^a=\ketbra{\theta^a_{\pm}}{\theta^a_{\pm}}$ and
$F_{\pm}^b=\ketbra{\theta^b_{\pm}}{\theta^b_{\pm}}$.

Although quantum theory predicts that arbitrarily small values of
$I_N$ can be obtained for large $N$, due to imperfections and errors
in the devices, it will not be possible to experimentally achieve
this.  In~\cite{suarez}, a discussion of the achievable values of
$I_N$ with imperfect visibilities was given.  For visibilities less
than 1, it is not optimal to take $N$ as large as possible to minimize
the observed $I_N$.  Thus, to get increasingly small bounds on the
variational distance in~\eqref{EQ:DBOUND}, one must increase the
experimentally obtained visibilities as well as the number of
measurement settings (see Figure~\ref{fig:visib}).



\section*{Acknowledgements}
We are grateful to {\v{C}}aslav Brukner, Ad\'an Cabello, Jerry
Finkelstein, J\"urg Fr\"ohlich, Nicolas Gisin, Gian Michele Graf,
Adrian Kent, Jan-\AA ke Larsson, Llu\'is Masanes, Nicolas Menicucci,
Ognyan Oreshkov, Stefano Pironio, Rainer Plaga, Sofia Wechsler and
Anton Zeilinger for discussions on this work and thank L\'idia del Rio
for the illustrations.  This work was supported by the Swiss National
Science Foundation (grant Nos.\ 200021-119868 and 200020-135048) and
the European Research Council (grant No.\ 258932).  Research at
Perimeter Institute is supported by the Government of Canada through
Industry Canada and by the Province of Ontario through the Ministry of
Research and Innovation.

\begin{widetext}


\makeatletter
\def\tagform@#1{\maketag@@@{(S\ignorespaces#1\unskip\@@italiccorr)}}
\makeatother
\numberwithin{equation}{subsection}

\def\bibsection{\section*{SUPPLEMENTARY REFERENCES}} 





\maketitle



\section*{SUPPLEMENTARY METHODS}

\section*{FORMAL STATEMENT OF THE CLAIM}

In this section, we provide a formal description of our result and the
assumptions it is based on. Their physical significance is explained
in the main text.

\subsection*{Definitions}

\begin{definition} A \emph{spacetime random variable} (\aSV), $X$, is
  a random variable
  together with a set of coordinates $(t, r_1, r_2, r_3) \in
  \mathbb{R}^4$.
\end{definition}

The coordinates can be used to define an order relation between
\aSV{}s, which one may interpret as a time ordering within
relativistic spacetime. (Note, however, that on a formal level, we do
not require any assumptions about relativity theory.)

\begin{definition}
  We say that a pair $(A, X)$ of \aSV{}s is \emph{time-ordered},
  denoted $A \bef X$, if the coordinate $(t, r_1, r_2, r_3)$ of $A$
  lies in the backward lightcone of the coordinate $(t', r'_1, r'_2,
  r'_3)$ of $X$, i.e., $(t-t')^2 \geq ||r-r'||^2$, $t \leq t'$.
  Furthermore, we say that two time-ordered pairs $A \bef X$ and $B
  \bef Y$ are \emph{spacelike separated} if $A \nbef Y$ and $B \nbef
  X$.
\end{definition}

The next two definitions refer to quantum theory or, more precisely,
quantum measurements. They will be used later for the formulation of
Assumption~\emph{\aQ{}}.

\begin{definition} \label{def:measurement}
  A \emph{quantum measurement}, denoted $(A \bef X,\,
  \{E_x^a\}_{a,x},\, \cH_S)$, is a pair of time-ordered \aSV{}s, $A
  \bef X$, called \emph{input} and \emph{output}, respectively,
  together with a family of measurement operators $\{E_x^a\}_{a,x}$ on
  a Hilbert space, $\cH_S$, such that $\sum_x
  (E_x^a)^{\dagger}E_x^a=\openone_{S}$ for all $a$.
\end{definition}

We interpret the input $A$ as the choice of an observable and $X$ as
the outcome of the measurement with respect to this
observable. Quantum theory determines the distribution of $X$
conditioned on $A$, depending on the quantum state $\rho_S$ of the
system to which the measurement is applied.

\begin{definition}\label{def:4}
  Given a density operator $\rho_S$ on $\cH_S$, the quantum
  measurement $(A \bef X,\, \{E_x^a\}_{a,x},\, \cH_S)$ is said to be
  \emph{compatible} with $\rho_S$ if
  \[
    P_{X|A}(x|a)=\tr((E_x^a)^{\dagger}E_x^a \rho_S) \ ,
  \]
  for all $a$ and $x$.  Likewise, a pair of quantum measurements $(A
  \bef X,\, \{E_x^a\}_{a,x},\, \cH_S)$ and $(B \bef Y,\,
  \{F_y^b\}_{b,y},\, \cH_T)$ is said to be \emph{compatible} with
  $\rho_{ST}$ defined on $\cH_S\ot\cH_T$ if
  \[
  P_{XY|AB}(xy|ab)=\tr\bigl[\bigl((E_x^a)^{\dagger}E_x^a\ot
  (F_y^b)^{\dagger}F_y^b\bigr)\rho_{ST}\bigr] \ ,
  \]
  for all $a$, $b$, $x$ and $y$.
\end{definition}

We describe the process of \emph{choosing} a value $A$ as a pair of
\aSV{}s, $O_A \bef A$, where $O_A$ is called the \emph{trigger event}
($O_A$ may be a constant). The process is considered \emph{free} if
the outcome $A$ is not correlated to anything that existed before the
trigger event $O_A$ in any reference frame.

\begin{definition}
  Given a set of \aSV{}s $\Gamma$, a \emph{free choice (with respect
    to $\Gamma$)} is a pair of time-ordered \aSV{}s, $O_A \bef A$,
  such that $A$ is statistically independent of the collection
  $\Gamma' := \{ W \in \Gamma: \, O_A \nbef W \}$, i.e., $P_{A
    \Gamma'} = P_{A} \times P_{\Gamma'}$. 
\end{definition}

\subsection*{Quantum-Mechanical Description of the Measurement Process}

Before stating our assumptions, let us briefly recall the
quantum-mechanical description of a measurement process. Most
generally, a quantum measurement on a system $S$ is described by a
family $\{E_x\}_x$ of operators acting on a Hilbert space $\cH_S$ such
that $\sum_x E_x^{\dagger} E_x = \openone$. If the state of $S$ before
the measurement is given by a density operator $\rho_S$ then each
possible outcome $X = x$ has probability
\begin{align*}
  P_X(x) = \tr( E_x^{\dagger}  E_x \rho_S) \ . 
\end{align*}
(Note that this is reflected by Definitions~\ref{def:measurement}
and~\ref{def:4}.)  Furthermore, conditioned on this outcome, the state
of $S$ after the measurement is
\begin{align*}
  \sigma_S^{(x)} = \frac{E_x \rho_S E_x^{\dagger}}{P_X(x)} \ .
\end{align*}
Averaged over all outcomes, the state is therefore given by
$\sigma_S=\cE(\rho_S)$, where $\cE$ is the trace-preserving
completely positive map (TPCPM) defined by
\begin{align*}
  \cE : \, \rho_S \mapsto \sigma_S= \sum_x P_X(x) \sigma_S^{(x)} = \sum_{x} E_x \rho_S E_x^{\dagger} \ .
\end{align*}

The TPCPM $\cE$ can be seen as part of an extended TPCPM $\bcE : \,
\rho_{S} \mapsto \sigma_{S D R}$ (in the sense that $\cE = \tr_{D R}
\circ \bcE$) which specifies the joint state $\sigma_{S D R}$ of $S$,
the measurement device, $D$, and possibly (parts of) the environment,
$R$, after the measurement (one may think of $\bcE$ as describing the
joint evolution that the system $S$, measurement device $D$ and the
environment $R$ undergo during a measurement).  By choosing a
sufficiently large environment, we can always take $\bcE$ to be an
isometry.  Since the measurement outcome $X$ is determined by the
final state of the measurement device $D$, there exists a family of
mutually orthogonal projectors $\{\Pi_x\}_x$ on the associated Hilbert
space $\cH_D$, where each $\Pi_x$ projects onto the subspace
containing the support of the state of $D$ corresponding to outcome $X
= x$. Formally, this corresponds to the requirement that
\begin{align} \label{eq:CPMdesc} \forall x: \, \tr_{DR} \bigl[ \bcE(\rho_{S}) (\openone_S \otimes \Pi_x \otimes \openone_R) \bigr]=  E_x \rho_{S}
 E_x^{\dagger} \ .
\end{align}

\subsection*{Assumptions}

To formulate our assumptions as well as our main claim, we consider an
arbitrary quantum measurement 
\begin{align} \label{eq:measurement}
  (A \bef X,\, \{E_x^{a}\}_{{a},x},\,\cH_S)
\end{align} 
with constant input $A=\bar{a}$ and output $X$. Furthermore, we consider
two \aSV{}s, $C$ and $Z$, such that $C\bef Z$, which model the access
to extra information provided by a potential extended theory.

\smallskip
\smallskip

Our first assumption demands that the measurement we consider is
correctly described by quantum mechanics.

\bigskip

\noindent \emph{{\bf Assumption~\emph{\aQa{}}.} There exists a pure
  quantum state $\rho_S$ which is compatible with the quantum
  measurement~\eqref{eq:measurement}.}

\bigskip

For the next assumption, let $\bcE: \, \rho_S \mapsto \sigma_{SDR}$ be
an isometry from $\cH_S$ to $\cH_S\ot\cH_D\ot\cH_R$ and let
$\{\Pi_x\}_x$ be a family of projectors such that~\eqref{eq:CPMdesc}
holds for the operators $\{E_x^{\bar{a}}\}_x$ specified by the
measurement~\eqref{eq:measurement}.\footnote{There are many ways to
  choose $\bcE$ and $\{\Pi_x\}_x$ with this property; our next
  assumption need only hold for one such choice.} The isometry $\bcE$
models the joint evolution of the system, $S$, on which the
measurement~\eqref{eq:measurement} is carried out, the measurement
device, $D$, and the parts of the environment, $R$, that may have been
affected by the measurement.\footnote{Note that,
  for~\eqref{eq:CPMdesc} to hold, it is sufficient that $\bcE$
  describes the interaction between $S$ and $D$ (and possibly $R$) on
  a microscopic scale and for a short time.  Hence, the fact that
  $\bcE$ is an isometry does not preclude subsequent ``collapse'' of
  the wave function.} We then consider arbitrary measurements
$\{F_x^{a}\}_{a,x}$ and $\{G_y^{b}\}_{b,y}$ on the subsystems $D$ and
$S R$, respectively, with the property that $F_x^{\bar{a}}=\Pi_x$.

The following assumption demands that the statistics produced by these
additional measurements are as predicted by quantum
theory. Furthermore, the outcome $X$ of the initial
measurement~\eqref{eq:measurement} can be recovered by measuring (in
an appropriate basis) the state of the device $D$ used for this
measurement.

\bigskip

\noindent \emph{{\bf Assumption~\emph{\aQb}.}  
  For appropriately defined \aSV{}s $A', X', B, Y$, the quantum
  measurements $(A' \bef X',\, \{F^a_x\}_{a,x},\, \cH_D)$ and $(B \bef
  Y,\, \{G_y^b\}_{b,y}, \cH_S\ot\cH_R)$ are compatible with $\sigma_{S
    D R} = \bcE(\rho_{S})$. Furthermore, the measurement on $D$ is
  consistent with the initial measurement~\eqref{eq:measurement}, in
  the sense that $X' = X$ whenever $A' = A = \bar{a}$.}

\bigskip

While the above assumptions are essentially consequences of the
requirement that the existing quantum theory is correct, our last
assumption demands that the measurement settings can be chosen freely.

\begin{center}
\includegraphics[width=0.4\textwidth]{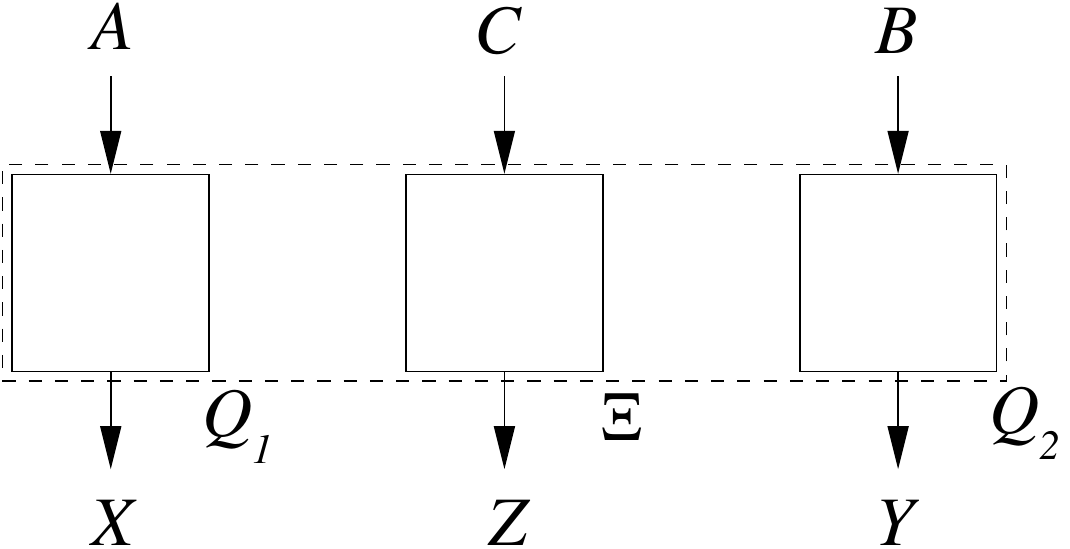}
\end{center} {\sf \textbf{Supplementary Figure S1 \!$|$\ \!Abstraction of the setup.} $Q_1$ and $Q_2$ depict a
pair of quantum systems with inputs $A$ and $B$ and outputs $X$ and
$Y$ respectively.  $\extension$ is a system which represents the
additional information provided by the extended theory.  Although
these three systems (solid boxes) can be independently manipulated,
they form parts of a larger system (dotted box).  While no restriction
is placed on the internal behaviour of the larger system, it follows
from Part~I of the proof that the combined distribution, $P_{X Y Z | A
  B C}$, is non-signalling.}\bigskip

\emph{\noindent {\bf Assumption~\emph{\aF}.}  There exist \aSV{}s
  $O_A$, $O_B$ and $O_C$ with $O_A \bef X'$, $O_B \bef Y$ and $O_C \bef
  Z$ spacelike separated such that $O_A \bef A'$, $O_B \bef B$ and
  $O_C \bef C$ are free choices with respect to $\{A',B,C,X',Y,Z\}$,
  and all possible values of $A'$ and $B$ are taken with nonzero
  probability.}

\subsection*{Main Claim}

\noindent {\bf Theorem~1.}  If the quantum
measurement~\eqref{eq:measurement}, modelled by the pair $A \bef X$,
and the additional information, $C \bef Z$, are such that
Assumptions~\emph{\aQa},~\emph{\aQb} and~\emph{\aF} are satisfied then
the Markov chain condition $X \leftrightarrow (A, C) \leftrightarrow
Z$ holds.


\section*{PART~II OF THE PROOF}

In this section, we prove the core inequality of Part~II of our proof,
Eqn.~8 in the Methods, which is stated as Lemma~\ref{lem:main}
below.

Recall the bipartite scenario described in the main text. The
measurements at each site are parameterized by values $A \in
\{0,2,\ldots,2N-2\}$ and $B \in \{1,3,\ldots,2N-1\}$ for some $N \in
\mathbb{N}$, and their respective outcomes, $X$ and $Y$, are taken to
be binary.  The measurements give rise to a joint probability
distribution $P_{XY|AB}$ from which we quantify the correlations
relevant for our statement in terms of $I_N$ defined by
\begin{align*}
I_N(P_{XY|AB}) := P(X = Y | A=0, B=2 N -1)\ + \sum_{\genfrac{}{}{0pt}{}{a,b}{|a-b|=1}} P(X \neq Y |A=a, B=b)  \ .
\end{align*}

We consider enlargements of this probability distribution, $P_{X Y Z | A
 B C}$ (see Figure~S1), that satisfy the non-signalling property (cf.\ Part~I of the proof), i.e.,
\begin{align}
P_{X Y | A B C} & = P_{X Y | A B} \label{eq:NS1} \\
P_{X Z | A B C} & = P_{X Z | A C} \label{eq:NS2} \\
P_{Y Z | A B C} & = P_{Y Z | B C} \label{eq:NS3} \ .
\end{align}
The claim is that any such extension approximately satisfies
$P_{Z|abcx} = P_{Z|abc}$, i.e., $Z$ is independent of $X$ for any
choices of $a$, $b$ and $c$. The accuracy of the approximation is measured
in terms of the variational distance.  For two distributions, $P_X$
and $P_Y$ over identical alphabets, this is defined by $D(P_X,P_Y) :=
\frac{1}{2}\sum_i|{P_X(i)-P_Y(i)}|$.

\begin{lemma}\label{lem:main}
  For any non-signalling probability distribution, $P_{XYZ|ABC}$, in
  which the random variables $X$ and $Y$ are binary, we have
\begin{align}\label{eq:markov}
D(P_{Z|abcx},P_{Z|abc})&\leq I_N(P_{XY|AB})
\end{align}
for all $a$, $b$, $c$, and $x$.
\end{lemma}

The proof is a generalization of an argument given
in~[15]
, which develops results of~[20] 
and~[24].

\begin{proof}
We first consider the quantity $I_N$ evaluated for the conditional
distribution $P_{X Y | A B, c z} = P_{X Y | A B C Z}(\cdot,
\cdot|\cdot, \cdot, c, z)$, for any fixed $c$ and $z$. The idea is
to use this quantity to bound the trace distance between the
conditional distribution $P_{X|a c z}$ and its negation, $ 1-P_{X|a
c z}$, which corresponds to the distribution of $X$ if its values
are interchanged.  If this distance is small, it follows that the
distribution $P_{X|a c z}$ is roughly uniform.

Let $P_{\bar{X}}$ be the uniform distribution on $X$. For $a_0 := 0$,
$b_0:= 2N-1$, we have
\begin{align}
I_N(P_{XY|AB,cz})
&=
P(X=Y|A=a_0,B=b_0,C=c,Z=z)+\sum_{\genfrac{}{}{0pt}{}{a,b}{|a-b|=1}}P(X\neq Y|A=a,B=b,C=c,Z=z) \nonumber   \\
&\geq D(1-P_{X|a_0 b_0 c z},P_{Y|a_0 b_0 c z}) +
\sum_{\genfrac{}{}{0pt}{}{a,b}{|a-b|=1}}D(P_{X|abcz},P_{Y|abcz})
\nonumber\\
&= D(1-P_{X|a_0cz},P_{Y|b_0cz}) +
\sum_{\genfrac{}{}{0pt}{}{a,b}{|a-b|=1}}D(P_{X|acz},P_{Y|bcz})
\nonumber\\
&\geq D(1-P_{X|a_0cz},P_{X|a_0cz})\nonumber \\
&= 2D(P_{X|a_0b_0cz}, P_{\bar{X}}) \label{ItoD} \ .
\end{align}
The first inequality follows from the fact that $D(P_{X|\Omega},
P_{Y|\Omega}) \leq P(X \neq Y | \Omega)$ for any event $\Omega$ (see
Lemma~\ref{lem:Dbound} below).  Furthermore, we have used the
non-signalling conditions $P_{X|abcz}=P_{X|acz}$ (from~\eqref{eq:NS2})
and $P_{Y|abcz}=P_{Y|bcz}$ (from~\eqref{eq:NS3}), and the triangle
inequality for $D$.  By symmetry, this relation holds for all $a$ and
$b$.  We hence obtain $D(P_{X|abcz},P_{\bar{X}})\leq\frac{1}{2}
I_N(P_{XY|AB,cz})$ for all $a$, $b$, $c$ and $z$.

We now take the average over $z$ on both sides of~\eqref{ItoD}. The
left-hand-side gives
\begin{align}
\sum_{z}P_{Z|abc}(z)I_N(P_{XY|AB,cz})&=\sum_{z}P_{Z|c}(z)I_N(P_{XY|AB,cz})\nonumber\\
&=\sum_{z}P_{Z|a_0b_0c}(z)P(X=Y|a_0, b_0, c,
z)+\sum_{\genfrac{}{}{0pt}{}{a,b}{|a-b|=1}}\sum_{z}P_{Z|abc}(z)P(X\neq
Y|a,b,c,z)\nonumber\\
&=P(X=Y|a_0,b_0,c)+\sum_{\genfrac{}{}{0pt}{}{a,b}{|a-b|=1}}P(X\neq
Y|a, b, c)\nonumber\\
&=I_N(P_{XY|AB,c}) \ ,
\end{align}
where we used the non-signalling condition $P_{Z|abc}=P_{Z|c}$ (which
is implied by~\eqref{eq:NS2} and~\eqref{eq:NS3}) several
times. Furthermore, taking the average on the right-hand-side
of~\eqref{ItoD} yields
$\sum_zP_{Z|abc}(z)D(P_{X|abcz},P_{\bar{X}})=D(P_{XZ|abc},P_{\bar{X}}\times
P_{Z|abc})$, so we have
\begin{equation}\label{eq:1}
2D(P_{XZ|abc},P_{\bar{X}}\times
P_{Z|abc})\leq I_N(P_{XY|AB,c})=I_N(P_{XY|AB}),
\end{equation}
where the last equality follows from the non-signalling
condition~\eqref{eq:NS1}.


Inequality~\eqref{eq:1} and the relation $D(P_X,Q_X)\leq
D(P_{XY},Q_{XY})$ imply
$D(P_{X|abc},P_{\bar{X}})\leq\frac{1}{2}I_N(P_{XY|AB})$, and hence
\begin{align}\label{eq:2}
\bigl|P_{X|abc}(x)-\frac{1}{2}\bigr|\leq\frac{1}{2}I_N(P_{XY|AB})
\end{align}
for all $a$, $b$, $c$ and $x$.
Furthermore, since
\begin{align*}
  2D(P_{XZ|abc},P_{\bar{X}}\times P_{Z|abc})=\sum_z\bigl|P_{XZ|abc}(0,z)-\frac{1}{2} P_{Z|abc}(z)\bigr|+\sum_z\bigl|P_{XZ|abc}(1,z)-\frac{1}{2} P_{Z|abc}(z)\bigr|,
\end{align*}
and both terms on the right-hand-side are equal, using~\eqref{eq:1} we have
\begin{align*}
  \sum_z \bigl|P_{XZ|abc}(x,z) - \frac{1}{2}P_{Z|abc}(z)\bigr|
  \leq\frac{1}{2}I_N(P_{XY|AB}),
\end{align*}
for all $a$, $b$, $c$ and $x$.  Combining this with~\eqref{eq:2} gives
\begin{align*}
  D(P_{Z|abcx},P_{Z|abc})
& =
  \sum_z \bigl|\frac{1}{2} P_{Z|abcx}(z) - \frac{1}{2} P_{Z|abc}(z) \bigr| \\
& \leq
  \sum_z \bigl|\frac{1}{2} P_{Z|abcx}(z) - P_{X|abc}(x) P_{Z|abcx}(z) \bigr|
  + \sum_z \bigl| P_{X|abc}(x) P_{Z|abcx}(z) - \frac{1}{2} P_{Z|abc}(z) \bigr| \\
& =
  \sum_z P_{Z|abcx}(z) \bigl| \frac{1}{2} - P_{X|abc}(x) \bigr|
   + \sum_{z} \bigl| P_{XZ|abc}(x,z) - \frac{1}{2} P_{Z|abc}(z) \bigr| \\
& \leq I_N(P_{X Y | A B}) \ .
\end{align*}
This establishes the relation~\eqref{eq:markov}.
\end{proof}

\begin{lemma} \label{lem:Dbound} Let $X$ and $Y$ be random variables
  jointly distributed according to $P_{X Y}$. The variational distance
  between the marginal distributions $P_X$ and $P_Y$ is bounded by
  \begin{align*}
    D(P_X, P_Y) \leq P(X \neq Y) \ .
  \end{align*}
\end{lemma}

\begin{proof}
  Let $P_{X Y}^{\neq} := P_{X Y | X\neq Y}$ be the joint distribution of
  $X$ and $Y$ conditioned on the event that they are not
  equal. Similarly, define $P_{X Y}^= := P_{X Y | X = Y}$. We then
  have
  \begin{align*}
    P_{X Y} = p_{\neq} P_{X Y}^{\neq} + (1-p_{\neq}) P_{X Y}^=
  \end{align*}
  where $p_{\neq} := P(X \neq Y)$.  By linearity, the marginals of these
  distributions satisfy the same relation, i.e.,
$$P_{X} = p_{\neq} P_{X}^{\neq} + (1-p_{\neq}) P_{X}^= \ \ \ \
\text{and}\ \ \ \ 
    P_{Y}   = p_{\neq} P_{Y}^{\neq} + (1-p_{\neq}) P_{Y}^=\, .$$
  Hence,  by the convexity of the variational distance,
  \begin{align*}
    D(P_X, P_Y) 
  \leq p_{\neq} D(P_X^{\neq}, P_Y^{\neq}) + (1-p_{\neq})  D(P_X^=,
  P_Y^=) 
  \leq p_{\neq} \ ,
  \end{align*}  
  where the last inequality follows because the variational distance
  cannot be larger than one, and $D(P_X^=, P_Y^=) = 0$.
\end{proof}

\section*{PART~III OF THE PROOF}

In this section we give the proof of the final part of
Theorem~1.\footnote{The proof we give here is similar to an argument
  given by Zurek~\cite{zurek} to derive the Born rule starting from
  unitarity.} We use the setup and assumptions as formulated at the
beginning of the Supplementary Methods.  In Parts~I and~II of the
proof (see the main text and the previous section) we showed that for
all $a$, $b$, $c$ and $x$, the relation $P_{Z|acx} = P_{Z|ac}$ holds
for projective quantum measurements compatible with one half of a
maximally entangled state (cf.\ Lemma~\ref{lem:main} and recall that
for such measurements, the quantity $I_N$ can be made arbitrarily
small for sufficiently large $N$).  Part~III, explained here, extends
this claim to arbitrary states (not necessarily maximally entangled
ones) and arbitrary measurements.

The argument proceeds in two steps. The first is to reduce the problem
to a situation where the measurement outcome is essentially
uniform. Let $(A \bef X,\, \{E_x^{\bar{a}}\}_{x},\, \cH_S)$ be the quantum
measurement under consideration (where the input $A=\bar{a}$ is
fixed). The idea is that we can always append a second measurement,
generating $\bar{X}$, such that the distribution of the joint output
$(X, \bar{X})$ is flat (to any desired accuracy).

\begin{lemma}
  Let $\varepsilon > 0$ and let $\rho_S$ be an arbitrary density
  operator on $\cH_S$.  For any measurement on $S$ there exists an
  additional measurement such that the joint output distribution of
  $(X,\bar{X})$, obtained by applying the two measurements
  sequentially to $\rho_S$, has distance $\varepsilon$ to a flat
  distribution.
\end{lemma}

\begin{proof}[Proof idea]
  It is easy to see that any probability distribution can be turned
  into an approximately flat one by adding an additional random
  process that ``splits'' each probability into sufficiently many
  smaller events. Furthermore, any such random process can be obtained
  by an appropriate choice of projective measurement (in a
  sufficiently large Hilbert space).  
\end{proof}

Let $\{E_{x,\bar{x}}^{\bar{a}}\}_{x,\bar{x}}$ be the set of measurement
operators corresponding to the measurement $(A \bef (X,\bar{X}),\,
\{E_{x,\bar{x}}^{\bar{a}}\}_{x,\bar{x}},\, \cH_S)$ which generates the pair $(X,
\bar{X})$, and let $\rho_S$ be a pure quantum state compatible with
this measurement (see Assumption~\emph{\aQa{}}).  Next, we introduce
projectors $\{\Pi_{x,\bar{x}}\}_{x,\bar{x}}$ and an isometry $\bcE$
such that $\sigma_{SDR}=\bcE(\rho_S)$ satisfies
\[
\tr_{DR}((\openone_S\ot\Pi_{x,\bar{x}}\ot\openone_R)\sigma_{SDR}(\openone_S\ot\Pi_{x,\bar{x}}\ot\openone_R))=E_{x,\bar{x}}^{\bar{a}}\rho_S(E_{x,\bar{x}}^{\bar{a}})^\dagger
\ .
\]
(Note that the isometry can always be defined such that the projectors
$\Pi_{x,\bar{x}}$ have rank one.) According to Assumption~\emph{\aQb{}}
we can append additional quantum measurements $(A'\bef
(X',\bar{X}'),\,\{F^a_{x,\bar{x}}\}_{a,x,\bar{x}},\,\cH_D)$ (with
$F^{\bar{a}}_{x,\bar{x}}=\Pi_{x,\bar{x}}$) and $(B\bef
Y,\,\{G^b_{y}\}_{b,y},\,\cH_S\ot\cH_R)$, such that the output
statistics are compatible with $\sigma_{S D R}$.  Furthermore,
$(X',\bar{X}')=(X,\bar{X})$ whenever $A'=\bar{a}$. Finally, by
Assumption~\emph{\aF} we can take $A'$ and $B$ to be free choices with
$O_A\bef (X',\bar{X}')$, $O_B\bef Y$, and $O_C\bef Z$ spacelike
separated (where $O_A$, $O_B$, and $O_C$ are the trigger events for
$A'$, $B$, and $C$, respectively).

Since the outcomes $(X',\bar{X}')$ of the measurement (for
$A'=\bar{a}$) are almost (up to an arbitrarily small distance
$\varepsilon$) uniformly distributed, and the state $\sigma_{SDR}$ is
pure, it must be (almost) maximally entangled between the measurement
device, $\cH_D$ and the remaining systems, $\cH_S\ot\cH_R$ (by a
suitable choice of the additional measurement, we can always take this
to be maximally entangled over an integer number of two-level
systems).  Furthermore, $\{\Pi_{x,\bar{x}}\}_{x,\bar{x}}$ are
orthogonal projectors. Hence, by a suitable choice of the additional
measurements producing $(X', \bar{X}')$ and $Y$, the argument given in
Parts~I and~II of the proof implies that, for any $\varepsilon>0$ and
for all $c$, $x$ and $\bar{x}$,
\begin{align*}
D(P_{Z|A'=\bar{a},cx \bar{x}},P_{Z|A'=\bar{a},c}) \leq \varepsilon \ .
\end{align*}
Since the values of $(X, \bar{X})$ and $(X', \bar{X}')$ coincide for
$A' = A = \bar{a}$ (cf.\ Assumption~\emph{\aQb}), we have
\begin{align*}
D(P_{Z|A=\bar{a},cx},P_{Z|A=\bar{a},c}) \leq \varepsilon \ .
\end{align*}
This relation holds for all $\bar{a}$, and, since $\varepsilon$ can be
arbitrarily small, establishes the desired Markov chain condition
$P_{Z|A=\bar{a},cx} =P_{Z|A=\bar{a},c}$.

\section*{REMARKS ON THE NOTION OF LOCALITY}

Here we make some comments about the notion of locality.  The main
point is to highlight that Bell's notion of locality is similar to,
but slightly less general than, the non-signalling nature of the
extension (as derived in Part~I of the proof).

To quote Bell~[2]
, \emph{locality} is the requirement that
``...the result of a measurement on one system [is] unaffected by
operations on a distant system with which it has interacted in the
past...''  Indeed, our non-signalling conditions reflect this
requirement and, in our language, the statement that $P_{XYZ|ABC}$ is
non-signalling is equivalent to a statement that the model is
\emph{local} (see also the discussion in~\cite{Hall}).  (We remind the
reader that we do not assume the non-signalling conditions, but
instead derive them from the free choice assumption.)

In spite of the above quote, Bell's formal definition of locality is
slightly more restrictive than these non-signalling conditions.  Bell
considers extending the theory using hidden variables, here denoted by
the variable $Z$.  He requires $P_{XY|ABZ}=P_{X|AZ}\times P_{Y|BZ}$
(see e.g.~[13]
), which corresponds to assuming not only $P_{X|ABZ}=P_{X|AZ}$ and
$P_{Y|ABZ}=P_{Y|BZ}$ (the non-signalling constraints, also called
\emph{parameter-independence} in this context), but also
$P_{X|ABYZ}=P_{X|ABZ}$ and $P_{Y|ABXZ}=P_{Y|ABZ}$ (also called
\emph{outcome-independence}).  These additional constraints do not
follow from our assumptions and are not used in this work.

A possible reason for the discrepancy is that Bell principally
considered extended theories which are deterministic given the hidden
variables.  In this case, the distinction between Bell's notion of
locality and the non-signalling conditions we use is unimportant: if
$X$ is deterministic given $A$ and $Z$, then $P_{X|ABYZ}=P_{X|AZ}$
follows automatically.  In fact, the converse also holds: given
parameter-independence and outcome-independence a necessary condition
for the model to recreate the quantum correlations arising from
measurements on a maximally entangled state is that it is
deterministic given the hidden variables.  To see this, note that for
any measurement $A=a$, there is a corresponding measurement $B=b_a$
such that quantum theory predicts identical outcomes.
In other words, $P_{X|ab_ayz}=\delta_{x,y}$.
The assumptions of parameter-independence and outcome-independence
give $P_{X|az}=P_{X|ab_ayz}$, and so $P_{X|az}(x)=\delta_{x,y}$.  This
implies that $X$ and $Y$ are determined given $A$ and $Z$.


\section*{CANDIDATE EXTENSIONS BASED ON SIMULATIONS OF QUANTUM CORRELATIONS}

It has been shown in a number of ways that quantum correlations can be
simulated from other resources.  For example, all correlations
generated by projective measurements on a maximally entangled pair of
qubits can be simulated by shared randomness and one bit of classical
communication~\cite{TonerBacon}, or by shared randomness and a
non-local box~\cite{CGMP} (a hypothetical device with
stronger-than-quantum correlations~\cite{Cirelson93,PR}). Furthermore,
these results have been generalized to arbitrary (not necessarily
maximally entangled) pure states~\cite{BGPS08}.  

Since such simulations recreate quantum correlations, they may appear
at first sight to be extensions of quantum theory.  We will not
provide an exhaustive treatment of all such models, but instead give a
short explanation as to why the examples above do not contradict our
claim.

First note that the ability to simulate quantum correlations does not
imply the ability to predict the outcomes of a genuine quantum
experiment.  However, when thinking about these simulations in the
context of extending quantum theory, the hypothesis is that the
components of the simulation really exist and are used to generate
outcomes.

The case where communication is needed is analogous to de Broglie-Bohm
theory~[16, 17] 
(discussed in the main text).  In order that the simulation can work
in the case of spacelike separated measurements, the communication
bit, $Z$ (which depends on one of the measurement choices, say $A$),
must propagate faster than light.  The bit $Z$ is therefore accessible
outside the future lightcone of $A$.  According to
Assumption~\emph{\aF}, it must be possible to choose $A$ to be
independent of this (now pre-existing) information, which would no
longer be the case.  Such models therefore contradict
Assumption~\emph{\aF}.

In the model of~\cite{CGMP}, where a non-local box is used for the
simulation, even with full access to this box, there is no better way
to predict the measurement outcomes.  To see this, note that the
output, $X$, of a measurement specified by a parameter, $A$, is
generated in the simulation by {\sc xor}ing a shared classical value
with the output of a non-local box, whose input depends on $A$.  Since
the individual outputs of a non-local box are uniform and random the
same is true for $X$.  Hence, while the simulation recreates the
correct quantum correlations, it does not extend quantum theory in the
sense of providing any extra information about future measurement
outcomes.  It is hence in agreement with Part~II of the proof

However, because it recreates the quantum correlations, the simulation
provides more information about the outcomes of joint measurements.
To see that this is incompatible with quantum theory, one would need
to apply Part~III of our argument, using a description of how the
model evolves under reversible operations.  Such a description is not
given in the above model and, furthermore, in consistent theories
which permit non-local boxes~\cite{barrett} the reversible dynamics
are known to be trivial~\cite{GMCD}.  They cannot therefore result in
a state whose statistics are consistent with those from a quantum
evolution, and hence contradict Assumption~\emph{\aQb}.



\end{widetext}

\end{document}